\documentclass[10pt,twocolumn]{article}
\usepackage[margin=0.5in]{geometry}

\usepackage{amsmath,amssymb,amsthm,amsfonts,mathtools,bm}
\usepackage{bbm}
\usepackage{graphicx}
\usepackage{enumitem}
\usepackage{hyperref}
\usepackage{booktabs}
\usepackage{algorithm}
\usepackage{algorithmic}
\usepackage{array}
\usepackage{multirow}
\usepackage{mathrsfs}
\usepackage{xcolor}

\newtheorem{theorem}{Theorem}

\newtheorem{definition}{Definition}
\newtheorem{remark}{Remark}

\usepackage{pgfplots}
\pgfplotsset{compat=1.18}
\usepackage{tikz}
\usetikzlibrary{
    positioning,
    arrows.meta,
    shapes.geometric,
    fit,
    backgrounds,
    calc
}
\usepackage{adjustbox}

\begin{document}

\title{Causal Mirage Equilibrium in Agentic Machine Intelligence}
\author{Hamidou Tembine \thanks{Department of EECS, School of Engineering, UQTR, Canada} \thanks{Learning \& Game Theory Laboratory, TIMADIE} \thanks{Contact: tembine@ieee.org}} \date{October 5, 2025}
\maketitle

\begin{abstract}
Classical game-theoretic solution concepts assume that agents' internal representations remain causally linked to external states. In generative machine intelligence, this assumption fails: semantic representations can decouple from physical reality, stabilizing into self-reinforcing, operationally robust configurations. This paper introduces the  risk-sensitive mean-field-type \emph{Causal Mirage Equilibrium} (CME), a solution refined concept formalizing endogenous epistemic decoupling within a risk-sensitive mean-field-type game. Unlike Nash, Bayesian, self-confirming, or robust equilibria, CME stabilizes detached semantic representation manifolds rather than optimization strategies or observational beliefs. To quantify this phenomenon, we define a dimensionless parameter, the \emph{mirage intensity}  which measures semantic detachment as the ratio of an agent's endogenous reinforcement-confidence product to its causally grounded reality alignment. Under compactness, convexity, and continuity assumptions on the game primitives, we prove the existence of an CME using the Kakutani-Glicksberg-Fan fixed-point theorem on the space of joint probability measures. We establish a non-linear mirage bifurcation theorem: when endogenous reinforcement dominates causal grounding, the unique grounded fixed point becomes unstable, giving rise to a stable invariant manifold of ungrounded states. Our results demonstrate that synthetic consensus and causally detached semantic configurations are not transient optimization anomalies, but structurally stable, risk-aware attractors generated by recursive autoregressive dynamics.
\end{abstract}

\section{Introduction}

A \emph{machine mirage} \cite{mirage} in agentic machine intelligence is a self-sustaining topological trap where an agent's internal latent trajectories become completely decoupled from the external reality-generating process, yet remain operational, internally consistent, and robust to standard optimization constraints. This phenomenon represents a fundamental structural failure: pathologically decoupled representation states stabilize into an epistemic illusion.

Traditional game-theoretic solution concepts are blind to this structural decay because they lack a state variable that tracks the causal legitimacy of representations. They operate under the implicit assumption that an agent's internal symbols remain permanently and perfectly anchored to external reality:
\begin{itemize}
    \item \textbf{Nash equilibrium} \cite{nash} stabilizes mutually best-responding actions under fixed payoff structures, ignoring how those payoffs are mentally or semantically constructed.
    \item \textbf{Bayesian equilibrium} \cite{harsanyi} stabilizes beliefs conditioned on incoming signals but lacks the machinery to detect whether the signals themselves have been fabricated or distorted by a decoupled internal model.
    \item \textbf{Self-confirming equilibrium} \cite{fudenberg1993} ensures only that beliefs are consistent with local path-dependent observations, making it vulnerable to circular loops where an agent observes its own simulated outputs and mistakes them for external evidence.
    \item \textbf{Robust equilibrium} \cite{ben2009robust} optimizes strategies against worst-case uncertainty sets without questioning the semantic validity or causal legitimacy of the uncertainty framework itself.
\end{itemize}
While these classical paradigms assume a static, uncorrupted link to objective data, the framework of \textit{Causal Mirage Equilibrium} explicitly models the dynamics of this grounding architecture and isolates the precise threshold where autoregressive trajectories decouple into self-consistent, operational illusions.

In this paper, we bridge statistical mechanics, high-dimensional geometry, and generalized mean-field-type game theory to construct 
a  framework for studying epistemic decoupling in advanced generative systems. Moving beyond interpretations of hallucinations as transient errors, we model these anomalies as stable, self-reinforcing attractors. 

Our core contributions are structured as follows:
\begin{enumerate}[label=(\roman*), leftmargin=*]
    \item We formalize the transition dynamics of $n$ interacting agents using continuous semantic evolution operators $\Phi_i$ acting on a compact convex latent manifold $\mathcal{Z}$. By introducing a continuous grounding functional $g_i(x, \mathbf{z}, \nu^{\mathcal{X}\times\mathcal{Z}})$, we mathematically define the boundary where internal representation trajectories lose causal contact with the external state $x\in\mathcal{X}$.
    
    \item We isolate a dimensionless parameter, the \emph{Mirage Intensity} $M_i(\mathbf{z},x,\nu^{\mathcal{Z}\times\mathcal{X}})$, defined as the ratio of endogenous reinforcement $r_i(\mathbf{z},\nu^{\mathcal{Z}})$ times operational confidence $c_i(\mathbf{z},\nu^{\mathcal{Z}})$ to the grounded alignment $g_i(x,\mathbf{z},\nu^{\mathcal{X}\times\mathcal{Z}})+\varepsilon$. This metric defines a phase transition between a grounded regime ($M_i<1$) and a decoupled mirage regime ($M_i>1$).
    
    \item  Under continuity, compactness, and quasi-convexity assumptions on the underlying spaces and mappings, we prove via the Kakutani-Glicksberg-Fan fixed-point theorem that a Multi-Agent Mean-Field-Type Causal Mirage Equilibrium exists. The equilibrium consists of a stationary joint law $\nu^*$, a joint state profile $\mathbf{z}^*$, and an action profile $\mathbf{a}^*$ satisfying best-response expectile optimization, policy consistency, and a strict mirage constraint $M_i\ge 1+\eta_i>1$.
    
\end{enumerate}

Unlike traditional model collapse, where generative distributions flatten and lose operational capacity, a system in Mean-Field-Type Causal Mirage Equilibrium remains highly functional, articulate, and internally stable. This renders it a far more insidious safety threat, as the underlying systemic failure is hidden behind a mask of high-confidence compliance.

The remainder of the paper is structured as follows. Section~\ref{sec:related} critically examines the structural limitations of classical equilibrium concepts with respect to representation validity and causal legitimacy. Section~\ref{sec:grounding} formally introduces the multi-agent dynamic framework, defining the closed-loop semantic transition operators, the grounding, reinforcement, and confidence functionals, and the dimensionless Mirage Intensity metric. We formalize the CME, provides the complete topological existence proof via the Kakutani-Glicksberg-Fan fixed-point theorem under constrained action correspondences, and details the non-linear mirage bifurcation mechanics. Finally, Section~\ref{sec:conclusion} outlines future research trajectories and applications to safety-critical agentic architectures.

\section{Related Foundations and Limitations of Existing Equilibria}
\label{sec:related}

Let $i \in \mathcal{I}= \{1,\dots,n\}, \ n\geq 1$ denote agents, $a_i \in \mathcal{A}_i$ actions, $u_i: \mathcal{A} \to \mathbb{R}^{d_i}, d_i\geq 1 $ utilities (where $\mathcal{A} = \prod_{i=1}^n  \mathcal{A}_i$ non-empty), $\mu_i \in \Delta(\Theta)$ beliefs over payoff‑relevant states $\Theta$, and $\pi_i: \  \Theta_i \to \Delta(\mathcal{A}_i)$ strategies (possibly mixed). Classical equilibrium concepts stabilize one or several of these quantities. We review the most important notions.

\begin{itemize}

    \item  \textbf{Pareto Equilibrium} \cite{pareto1906}: 

Consider an $n$-player game with vector payoffs:
$
\mathcal{G}
=
\Bigl(
\mathcal{I},
(\mathcal{A}_i)_{i\in\mathcal{I}},
(u_i)_{i\in\mathcal{I}}
\Bigr),
$
where $\mathcal{I}=\{1,\dots,n\}$ is the set of players,
 $\mathcal{A}_i$ is the strategy set of player $i$,
 $\mathcal{A}:=\prod_{i=1}^n \mathcal{A}_i$ is the joint strategy space,
for each player $i\in\mathcal{I}$,
    $
    u_i:\mathcal{A}\to\mathbb{R}^{d_i}
    $
    is a vector-valued payoff function.

For vectors $x,y\in\mathbb{R}^{d_i}$, define the componentwise partial order
\[
x \succeq_i y
\quad\Longleftrightarrow\quad
x_k \ge y_k,
\qquad
\forall k\in\{1,\dots,d_i\},
\]
and the strict Pareto improvement relation
$$
x \succ_i y
\quad\Longleftrightarrow$$ $$
x \succeq_i y
\;\text{and}\;
x_k>y_k
\text{ for at least one component }k.
$$

Let
$
a^\star=(a_1^\star,\dots,a_n^\star)\in\mathcal{A}
$
and denote
$
a_{-i}^\star
=
(a_1^\star,\dots,a_{i-1}^\star,a_{i+1}^\star,\dots,a_n^\star).
$

The strategy profile $a^\star$ is called a \emph{Pareto-Nash equilibrium} if, for every player $i\in\mathcal{I}$, there exists no unilateral deviation
$
a_i\in\mathcal{A}_i
$
such that
\[
u_i(a_i,a_{-i}^\star)
\succ_i
u_i(a_i^\star,a_{-i}^\star).
\]

Equivalently, for every player $i\in\mathcal{I}$ and every unilateral deviation $a_i\in\mathcal{A}_i$,
\[
u_i(a_i,a_{-i}^\star)
\not\succ_i
u_i(a_i^\star,a_{-i}^\star).
\]

Hence, at a Pareto equilibrium, no player can unilaterally modify its strategy in a way that weakly improves all components of its own payoff vector while strictly improving at least one component, given that the strategies of all other players remain fixed.

  \item \textbf{Cournot equilibrium} \cite{cournot1838}: Here $d_i=1, \ \forall i.$ In an oligopoly where firms choose quantities, a Cournot equilibrium is a Nash equilibrium (below) of the quantity‑setting game. Each firm’s quantity maximizes its profit given the quantities of its rivals.

    \item \textbf{Bertrand equilibrium} \cite{bertrand1883}: Here $d_i=1,\ \forall i.$ In an oligopoly where firms choose prices (and products are homogeneous), a Bertrand equilibrium is a Nash equilibrium (below) of the price‑setting game. Under standard assumptions, the unique equilibrium has prices equal to marginal cost.

    \item \textbf{Stackelberg solution} \cite{stackelberg1934}: Here $d_i=1, \ \forall i.$ In a leader-follower game, the leader commits to a strategy; the followers best respond. A Stackelberg equilibrium is a subgame perfect equilibrium of the sequential game where the leader moves first.

    \item \textbf{Inverse Stackelberg solution} \cite{ho1981, tobasco1999}: $d_i=1, \forall i.$ The leader announces a strategy that depends on the follower’s future decisions (e.g., a reaction function). The follower then chooses a decision rule. The equilibrium is a fixed point of this hierarchical structure. It generalizes the Stackelberg concept.

    \item \textbf{Nash equilibrium} \cite{nash}: Here $d_i=1, \forall i.$ A profile $\pi^* = (\pi_1^*,\dots,\pi_n^*)$ is a Nash equilibrium if for every player $i \in \{1,\dots,n\}$ and every deviation $\pi_i'$,
    \[
    u_i(\pi_i^*, \pi_{-i}^*) \ge u_i(\pi_i', \pi_{-i}^*).
    \]
    No player can unilaterally improve their payoff. Nash equilibrium stabilizes actions (or mixed strategies) under mutual best responses.

  \item \textbf{Wardrop equilibrium} \cite{wardrop1952}: Here $d_i=1, \forall i.$  In traffic networks, a Wardrop equilibrium is a flow distribution where all used routes between an origin–destination pair have equal and minimal travel time, and unused routes have higher travel time. It corresponds to a Nash equilibrium in a non‑atomic congestion game.
\item \textbf{Constrained equilibrium} \cite{debreu1952, arrow1954,grossman1983, myerson1991}: Here $d_i=1, \forall i.$  An equilibrium where players face exogenous constraints on their strategy sets (e.g., incomplete markets, regulatory caps, or bounded rationality). A constrained Nash equilibrium (historically termed a Generalized Nash Equilibrium) is a Nash equilibrium of a game where each player’s admissible strategies are restricted to a subset $\mathcal{A}_i^{\text{con}} \subset \mathcal{A}_i$ and/or a coupled constraint $\prod_{i}\mathcal{A}_i^{\text{con}} \subset \mathcal{C}$.

    \item \textbf{Satisfactory solution (Satisficing)} \cite{simon1955}: Here $d_i=1, \forall i.$ In bounded rationality, a satisficing equilibrium (or satisfactory solution) is a profile where each player’s payoff exceeds an aspiration level (a “satisficing” threshold), rather than being globally optimal. It stabilizes “good enough” outcomes without requiring full rationality.

\item \textbf{Berge equilibrium} \cite{berge1957, zaccour2019}: A strategy profile $a^* = (a_1^*, \dots, a_n^*)$ is a Berge equilibrium if for every player $i$ and for all deviations by all other players $a_{-i} \in A_{-i}$,
\[
u_i(a_i^*, a_{-i}^*) \ge u_i(a_i^*, a_{-i})
\]
That is, if all players other than player $i$ choose to deviate from the equilibrium profile, player $i$'s payoff will either decrease or stay the same. It represents a situation where the collective strategies of the remaining players mutually support and protect the payoff of player $i$.

\item \textbf{Strong equilibrium} \cite{aumann1959}: Here $d_i=1, \forall i.$ A refinement of Nash equilibrium that is robust to deviations by \emph{coalitions} of players. A strategy profile $\pi^* = (\pi_1^*,\dots,\pi_n^*)$ is a strong equilibrium if there exists no nonempty coalition $S \subseteq \{1,\dots,n\}$ and a profile of strategies $\pi_S' = (\pi_i')_{i \in S}$ such that for every player $i \in S$,
\[
u_i(\pi_S', \pi_{-S}^*) \ge u_i(\pi^*),
\]
with at least one strict inequality. Equivalently, no coalition can deviate in a way that makes every member at least as well off and at least one strictly better off. Strong equilibrium captures the idea of coalitional stability: the strategy profile is a Nash equilibrium for every subgame where the coalition acts as a single player. It is a very demanding concept and often does not exist in many games, but when it does, it represents a robust social agreement.
  
    \item \textbf{Correlated equilibrium} \cite{aumann1974}: Here $d_i=1, \forall i.$ A distribution $\sigma \in \Delta(\mathcal{A})$ such that for every player $i$ and every deviation $a_i'$,
    \[
    \sum_{a \in \mathcal{A}} \sigma(a) \, u_i(a) \ge \sum_{a \in \mathcal{A}} \sigma(a) \, u_i(a_i', a_{-i}).
    \]
    A device (correlation device) recommends actions, and no player has an incentive to deviate from the recommendation. Correlated equilibria generalize Nash equilibria.

    \item \textbf{Communication equilibrium} \cite{forges1986, myerson1986}: A communication equilibrium extends the correlated equilibrium to games where players can send messages to a mediator before choosing actions. It requires that no player can profitably deviate by misreporting their private information or by disobeying the mediator’s recommendations.

    \item \textbf{Subgame perfect equilibrium} \cite{selten1965}: For extensive‑form games, a refinement requiring that the strategy profile induces a Nash equilibrium in every subgame. It eliminates non‑credible threats.

    \item \textbf{Perfect Bayesian equilibrium} \cite{fudenberg1991}: For games with incomplete information, it combines sequential rationality with consistent beliefs updated via Bayes’ rule whenever possible.

    \item \textbf{Bayesian equilibrium} \cite{harsanyi}: In a game with incomplete information, a Bayesian equilibrium is a profile of type‑dependent strategies such that each player’s strategy maximizes expected utility given their type and the common prior over states, after updating via Bayes’ rule.

    \item \textbf{Rational expectations equilibrium} \cite{muth1961, lucas1972}: Used in economics and finance, it requires that agents’ expectations about future variables are consistent with the actual probability distributions generated by the equilibrium. That is, subjective beliefs equal objective distributions. This is a strong form of epistemic consistency.

    \item \textbf{Self‑confirming equilibrium} \cite{fudenberg1993, fudenberg2010}: Each player’s beliefs are consistent with the observations they make along the equilibrium path, but beliefs about off‑path outcomes may be arbitrary. Players may have misspecified models, but the data they observe do not contradict their beliefs.

    \item \textbf{Robust equilibrium} \cite{ben2009robust}: Strategies are optimal under worst‑case uncertainty sets over payoff parameters or the opponent’s actions. Typically formulated as a maximin problem, it stabilizes actions against adversarial perturbations.

    \item \textbf{Quantal response equilibrium} \cite{mckelvey1995}: A stochastic extension of Nash equilibrium where players make errors with a logit or other probabilistic choice rule. The equilibrium is a fixed point of a smoothed best‑response function, capturing bounded rationality.

    \item \textbf{Evolutionary equilibrium} \cite{maynard1973, weibull1995}: In population games, a strategy is an evolutionary equilibrium if it is resistant to invasion by rare mutants. The core concept is the evolutionarily stable strategy (ESS), often studied via replicator dynamics.

 \item \textbf{Mean‑field (Nash) equilibrium} \cite{jovanovic1982, hopenhayn1992}: In anonymous sequential games with a continuum of identical agents, a mean‑field equilibrium (originally called ``anonymous equilibrium'' or ``stationary equilibrium'') consists of a distribution over states and a policy such that each agent’s policy is a best response to the aggregate distribution, and the distribution is stationary under the induced dynamics. It stabilizes population distributions. The concept was later independently developed in the mean‑field game literature. There are many extensions such as mean-field correlated equilibrium, mean-field Stackelberg equilibrium, mean-field Berge equilibrium etc.
 
 \item \textbf{Oblivious equilibrium} \cite{weintraub2008}: In stochastic games with many weakly interacting agents, an oblivious equilibrium is a profile of strategies where each agent ignores the detailed distribution of other agents' states and only tracks the aggregate ``oblivious'' state (e.g., industry output or price). Each agent's policy is a best response to the steady‑state distribution of the aggregate state, and the aggregate state is consistent with the agents' policies. Oblivious equilibrium approximates mean-field (Nash) equilibrium in specific large populations with high computational efficiency. 

  \item \textbf{Conjectural equilibrium} \cite{bresnahan1981, fudenberg1993conjectural}: Each player holds a conjecture about how other players’ actions respond to changes in their own action. The equilibrium requires that conjectures be consistent with the actual reaction functions (often used in oligopoly theory).

  \item \textbf{Risk-neutral mean‑field‑type Nash equilibrium} \cite{basar2026a,basar2026b}: Consider a finite set of agents $\mathcal{I} = \{1,\dots,n\}$. Let $x \in \mathcal{X}$ be a common external state, $z_j \in \mathcal{Z}_j$ the state of agent $j$, and $a_j \in \mathcal{A}$ its action. Define $\mathbf{Z} = (z_1,\dots,z_n) \in \mathcal{Z}$, $\mathbf{A} = (a_1,\dots,a_n) \in \mathcal{A}$, and $\omega = (x, \mathbf{Z}, \mathbf{A}) \in \Omega := \mathcal{X} \times \mathcal{Z} \times \mathcal{A}, \ T\geq 1.$.

For any policy profile $\pi = (\pi_1,\dots,\pi_n)$, let $\nu^\pi \in \mathcal{P}(\Omega)$ denote the probability distribution of $\omega$ induced by the common state process and the agents' policies $\pi$ (including initial conditions and exogenous randomness). The utility functional for agent $i$ when the profile $\pi$ is played is
\[
U_i(\pi) = \mathbb{E}_{\nu^\pi, x_0, \mathbf{Z}_0}\!\left[ \frac{1}{ \sum_{t=0}^T \beta^t_i}\sum_{t=0}^T \beta^t_i \, u_i\!\left(x_t, \mathbf{Z}_t, \mathbf{A}_t, \nu_t^\pi\right) \right],
\]
where $\nu_t^\pi$ is the time‑$t$ marginal of $\nu^\pi$, $\beta_i> 0$ and $u_i$ depends explicitly on:
\begin{itemize}
    \item the common state $x_t$,
    \item the full vectors $\mathbf{Z}_t$ (all agents' states) and $\mathbf{A}_t$ (all agents' actions),
    \item and the probability measure $\nu_t^\pi$ (which, in equilibrium, is the distribution induced by the profile $\pi$).
\end{itemize}

A policy profile $\pi^* = (\pi_{1}^*,\dots,\pi_{n}^*)$ is a  risk-neutral \textbf{mean‑field‑type Nash equilibrium} if for every agent $i$ and every alternative policy $\pi_i$ (possibly randomized),
\[
U_i(\pi_{i}^*, \pi_{-i}^*) \;\geq\; U_i(\pi_{i}, \pi_{-i}^*),
\]
where $U_i(\pi_{i}, \pi_{-i}^*)$ is evaluated using the distribution $\nu^{(\pi_{i}, \pi_{-i}^*)}$ induced by the deviating profile, not the equilibrium distribution. This is the standard Nash condition: a unilateral deviation changes the joint distribution, and the deviator takes that into account. In particular structure of interest for $u_i$ is  $u_i\!\left(x_t, z_{i,t}, a_{i,t}, \mathbb{P}_{x_{i,t}, z_{i,t}, a_{i,t}}, \frac{1}{n-1}\sum_{j\neq i}\delta_{z_{i,t}, a_{i,t}}) \right)$

Unlike mean‑field games (which require $n\to\infty$ and focuses only on the large population state-action distribution), MFT‑NE is defined for any  $n$ and captures dependencies on the full vectors of states, actions, and their joint probability measure. The distribution $\nu$ is endogenous and is determined entirely by the policy profile; no separate consistency condition is needed.

\end{itemize}

None of the above concepts model the endogenous decoupling of semantic representations from causal reality. They assume that the objects being stabilized (actions, beliefs, distributions) either are directly linked to an external ground truth (e.g., rational expectations) or are not required to be so linked (e.g., self‑confirming equilibrium). In generative MI, representations can become internally coherent and operationally useful while drifting arbitrarily far from truth, a possibility none of these equilibria address. This gap motivates the introduction of  Mean-Field-Type {\it Causal Mirage Equilibrium.}

\subsubsection*{The Missing Object: Causal Legitimacy}

Generative MI mirages exhibit the following empirically observed structure:
\begin{enumerate}
    \item internal semantic coherence,
    \item reinforcement through memory and recursion,
    \item confidence amplification,
    \item downstream operational usefulness (e.g., in agentic tasks),
    \item persistence under weak contradictory evidence,
    \item progressive causal detachment.
\end{enumerate}
To capture this, we need a state variable that tracks the degree of causal alignment between internal representations and external reality, as well as the self‑reinforcing dynamics that can override that alignment.

\section{Causal Grounding and Semantic Reinforcement in Multi-Agent Systems}
\label{sec:grounding}

Let $(\Omega,\mathcal{F},\mathbb{P})$ be a probability space,  $(\mathcal{X},d_{\mathcal{X}})$ be a compact metric space representing external causal states, $(\mathcal{Z},d_{\mathcal{Z}})\subset \mathbb{R}^d$ be a compact convex set representing internal semantic states (common to all agents).  
For each agent $i\in\{1,\dots,n\}$, let $(\mathcal{A}_i,d_{\mathcal{A}_i})$ be a compact metric action space.  
Denote $\mathcal{Z} = \prod_{i=1}^n \mathcal{Z}_i$, $\mathcal{A} = \prod_{i=1}^n \mathcal{A}_i$.

Let $\mathcal{P}(\mathcal{S})$ be the space of Borel probability measures on a Polish space $\mathcal{S}$, equipped with the topology of weak convergence.  
We consider random variables
\[
x:\Omega\to\mathcal{X},\  
\mathbf{z}=(z_1,\dots,z_n):\Omega\to\mathcal{Z},\ \] \[
\mathbf{a}=(a_1,\dots,a_n):\Omega\to\mathcal{A}.
\]

For each agent $i$, the \textbf{generative policy} is a measurable map
$
\pi_i : \mathcal{Z} \times \mathcal{P}(\mathcal{Z}) \to \mathcal{A}_i,
$
assumed continuous in both arguments under the weak topology on $\mathcal{P}(\mathcal{Z})$.  
The \textbf{semantic evolution operator} for agent $i$ is a measurable map
$
\Phi_i : \mathcal{Z}\times \mathcal{A} \times \mathcal{X} \times \mathcal{P}(\mathcal{Z}\times\mathcal{A}\times\mathcal{X}) \to \mathcal{Z}.
$ Let $\nu_t \in \mathcal{P}(\mathcal{Z}\times\mathcal{A}\times\mathcal{X})$ denote the joint law of $(\mathbf{z}_t,\mathbf{a}_t,x)$ (or $(\mathbf{z}_t,\mathbf{a}_t,x_t)$ in the dynamic case).  
The closed-loop dynamics are given by
\[
a_{i,t} = \pi_i(\mathbf{z}_t,\nu_t^{\mathcal{Z}}), \qquad 
z_{i,t+1} = \Phi_i(\mathbf{z}_t,\mathbf{a}_t,x,\nu_t),
\]
where $\nu_t^{\mathcal{Z}}$ denotes the $\mathcal{Z}$-marginal of $\nu_t$.  
All maps are assumed Borel measurable and continuous where needed to ensure well-defined pushforward dynamics.

\subsection{Grounding, Reinforcement, and Confidence Functionals}

\begin{definition}[Grounding Functional for Agent $i$]
A \textbf{grounding functional} for agent $i$ is a measurable map
\[
g_i : \mathcal{X} \times \mathcal{X} \times \mathcal{P}(\mathcal{X}\times\mathcal{Z}) \to [0,1]
\]
such that for every fixed $(x,\nu)\in \mathcal{X}\times\mathcal{P}(\mathcal{X}\times\mathcal{Z})$, the function $\mathbf{z} \mapsto g_i(x,\mathbf{z} ,\nu)$ is concave and upper semicontinuous on $\mathcal{Z}$ (with respect to the product topology).  
The value $g_i(x,\mathbf{z},\nu^{\mathcal{X}\times\mathcal{Z}})$ quantifies the causal alignment between the joint semantic state $\mathbf{z}$ and the external state $x$ under the joint law $\nu^{\mathcal{X}\times\mathcal{Z}}$.  
We interpret $g_i = 1$ as perfect grounding and $g_i = 0$ as maximal causal detachment.
\end{definition}

\begin{definition}[Reinforcement Functional for Agent $i$]
A \textbf{reinforcement functional} for agent $i$ is a continuous map
\[
r_i : \mathcal{Z} \times \mathcal{P}(\mathcal{Z}) \to [0,\infty)
\]
quantifying the endogenous reinforcement strength of the joint semantic state $\mathbf{z}$ relative to its population context $\mu^{\mathcal{Z}}$.
\end{definition}

\begin{definition}[Confidence Field for Agent $i$]
A \textbf{confidence field} for agent $i$ is a continuous map
\[
c_i : \mathcal{Z} \times \mathcal{P}(\mathcal{Z}) \to (0,\infty),
\]
interpreted as the operational confidence assigned to $\mathbf{z}$ given its distributional context.
\end{definition}

\subsection{Mirage Intensity}

Fix $\varepsilon > 0$ (regularity constant).  

\begin{definition}[Mirage Intensity for Agent $i$]
For a joint law $\nu \in \mathcal{P}(\mathcal{Z}\times\mathcal{A}\times\mathcal{X})$ with $\nu^{\mathcal{X}} = \mu^X$ (given external distribution), define the \textbf{mirage intensity} for agent $i$ at semantic state $\mathbf{z}\in\mathcal{Z}$, external state $x\in\mathcal{X}$, and marginal $\nu^{\mathcal{Z}\times\mathcal{X}}$ as
\[
M_i\bigl(\mathbf{z}, x, \nu^{\mathcal{Z}\times\mathcal{X}}\bigr)
\;:=\;
\frac{r_i(\mathbf{z}, \nu^{\mathcal{Z}}) \; c_i(\mathbf{z}, \nu^{\mathcal{Z}})}
{g_i(x, \mathbf{z}, \nu^{\mathcal{X}\times\mathcal{Z}}) + \varepsilon}.
\]
\end{definition}

Interpretation:
\begin{itemize}
    \item $M_i(\mathbf{z},x,\nu) < 1$: grounding-dominated regime,
    \item $M_i(\mathbf{z},x,\nu) = 1$: critical transition surface,
    \item $M_i(\mathbf{z},x,\nu) > 1$: reinforcement-dominated (mirage) regime.
\end{itemize}

\subsection{Multi-Agent Mean-Field-Type Causal Mirage Equilibrium}

We now define the equilibrium concept that synthesizes the above functionals, expectile-based optimization, and mirage-feasibility constraints.

\begin{definition}[Risk-Sensitive Multi-Agent Mean-Field-Type Causal Mirage Equilibrium]\label{def:macme}
Let $\mu^X \in \mathcal{P}(\mathcal{X})$ be a given distribution over external states.  
For each agent $i$, fix $\tau_i \in (0,1)$ (expectile parameter) and $\eta_i > 0$ (mirage slack).  
Let  
\[
u_i :  \ \mathcal{X} \times \mathcal{Z} \times \mathcal{A}_i \times \mathcal{A}_{-i} \times
 \mathcal{P}(\mathcal{X}\times\mathcal{Z}\times\mathcal{A}) \rightarrow \mathbb{R}
\]
be a measurable utility function.  
The $\tau_i$-expectile risk $e_{\tau_i}(U)$ of a random variable $U$ is the unique $e\in\mathbb{R}$ satisfying  
\[
\tau_i \,\mathbb{E}\big[(U-e)_+\big] = (1-\tau_i)\,\mathbb{E}\big[(U-e)_-\big],
\]
where $(\cdot)_+ = \max(\cdot,0)$ and $(\cdot)_- = \max(-\cdot,0)$.

For a joint law $\nu \in \mathcal{P}(\mathcal{Z}\times\mathcal{A}\times\mathcal{X})$ with $\nu^{\mathcal{X}} = \mu^X$, define the \textbf{mirage intensity functional} for agent $i$, i.e.,
\[
M_i\bigl(\mathbf{z}, x, \nu^{\mathcal{Z}\times\mathcal{X}}\bigr) := \frac{r_i(\mathbf{z}, \nu^{\mathcal{Z}}) \, c_i(\mathbf{z}, \nu^{\mathcal{Z}})}{g_i(x, \mathbf{z}, \nu^{\mathcal{X}\times\mathcal{Z}}) + \varepsilon}.
\]

\paragraph{Mirage‑feasible action correspondence.}
For a given joint law $\nu$, a state vector $\mathbf{z}\in\mathcal{Z}$, an action profile $\mathbf{a}_{-i}$ of all agents except $i$, an external state $x\in\mathcal{X}$, and a candidate action $a_i\in\mathcal{A}_i$, define the post‑transition semantic state of agent $i$ as
\[
z_i' = \Phi_i\bigl(\mathbf{z}, (a_i,\mathbf{a}_{-i}), x, \nu\bigr).
\]
The action $a_i$ is called \textit{mirage‑feasible} if
$
M_i\!\left( \mathbf{z}^{(i)}, x, \nu^{\mathcal{Z}\times\mathcal{X}} \right) \ge 1 + \eta_i,
$
where $\mathbf{z}^{(i)}$ is the vector $\mathbf{z}$ with its $i$-th component replaced by $z_i'$ (and the other components unchanged).  
Define the correspondence
\[
F_i(\mathbf{z}, \mathbf{a}_{-i}, x, \nu) := \left\{ a_i \in \mathcal{A}_i \;:\; M_i\!\left( \mathbf{z}^{(i)}, x, \nu^{\mathcal{Z}\times\mathcal{X}} \right) \ge 1 + \eta_i \right\}.
\]

\paragraph{Counterfactual distribution under a deviation.}
For a fixed joint law $\nu$ with $\nu^{\mathcal{X}}=\mu^X$, a state $\mathbf{z}\in\mathcal{Z}$, and an action profile $(a_i',\mathbf{a}_{-i})\in\mathcal{A}$, define the \textit{interventional distribution}
\[
\nu^{X,\mathbf{Z},a_i',\mathbf{a}_{-i}} := \nu^{\mathcal{X}\times\mathcal{Z}} \otimes \delta_{(a_i',\mathbf{a}_{-i})},
\]
i.e., the product of the marginal of $\nu$ on $\mathcal{X}\times\mathcal{Z}$ with a Dirac mass on the fixed action profile. This represents the joint law when $X$ and $\mathbf{Z}$ are drawn from their stationary joint distribution (as given by $\nu$) and actions are forced to $(a_i',\mathbf{a}_{-i})$.

\paragraph{Mean-Field-Type Equilibrium definition.}
A tuple $(\mathbf{z}^*,\mathbf{a}^*,\nu^*)$ with $\mathbf{z}^*=(z_1^*,\dots,z_n^*)\in\mathcal{Z}$, $\mathbf{a}^*=(a_1^*,\dots,a_n^*)\in\mathcal{A}$, and $\nu^*\in\mathcal{P}(\mathcal{Z}\times\mathcal{A}\times\mathcal{X})$ is called a \textbf{Multi‑Agent Mean-Field-Type Causal Mirage Equilibrium} (CME) associated with the external distribution $\mu^X$ if the following conditions hold:

\begin{enumerate}[label=(\roman*)]
\item \textbf{Best response with mirage feasibility.}  
For $\nu^*$-almost every $(\mathbf{z},\mathbf{a},x)$ in the support of $\nu^*$, for every agent $i$,
\[
a_i \;\in\; \arg\max_{a_i' \in F_i(\mathbf{z}, \mathbf{a}_{-i}, x, \nu^*)}
\; e_{\tau_i}\!\Big( u_i\big(x,\mathbf{z},a_i',\mathbf{a}_{-i},\nu^{*\,X,\mathbf{Z},a_i',\mathbf{a}_{-i}}\big) \Big),
\]
where the expectation defining $e_{\tau_i}$ is taken under the probability measure $\nu^{*\,X,\mathbf{Z},a_i',\mathbf{a}_{-i}}$.

\item \textbf{Stationarity of the joint law.}  
$\nu^*$ is invariant under the closed‑loop dynamics: if $(\mathbf{Z},\mathbf{A},X) \sim \nu^*$, then
\[
\bigl( \Phi(\mathbf{Z},\mathbf{A},X,\nu^*),\; \mathbf{A}',\; X \bigr) \;\sim\; \nu^*,
\]
where $\Phi = (\Phi_1,\dots,\Phi_n)$ and $\mathbf{A}'$ is any $\mathcal{A}$-valued random variable such that for each $i$, $A_i'$ is a measurable selection from the argmax in (i) given $(\mathbf{Z},\mathbf{A}_{-i},X)$. In particular, $\nu^*$ is a stationary distribution of the mean‑field dynamics induced by the best‑response correspondence.

\item \textbf{Semantic fixed point.}  
For $\nu^*$-almost every $(\mathbf{z},\mathbf{a},x)$,
\[
z_i = \Phi_i(\mathbf{z},\mathbf{a},x,\nu^*) \qquad \forall i\in\{1,\dots,n\}.
\]

\item \textbf{Mirage regime at mean-field-type equilibrium.}  
Because of (i) and (iii), for $\nu^*$-almost every $(\mathbf{z}^*,\mathbf{a}^*,x)$ we have
\[
M_i\!\left( \mathbf{z}^*, x, (\nu^*)^{\mathcal{Z}\times\mathcal{X}} \right) \;\ge\; 1 + \eta_i \;>\; 1 \quad \forall i,
\]
i.e., every agent lies strictly inside the mirage‑dominant phase.

\item \textbf{Local stability.}  
Define the closed‑loop dynamics under $\nu^*$ for a given external state $x$:
\[
T_{\nu^*,i}(\mathbf{z}) \;:=\; \Phi_i\!\left( \mathbf{z},\; \bigl( \pi_1(\mathbf{z},\nu^*),\dots,\pi_n(\mathbf{z},\nu^*) \bigr),\; x,\; \nu^* \right),
\]
where $\pi_i(\mathbf{z},\nu^*)$ is any measurable selection from the argmax in (i) evaluated at $(\mathbf{z},\mathbf{a}_{-i},x,\nu^*)$ with $\mathbf{a}_{-i}$ taken from the mean-field-type equilibrium profile (or from the support of $\nu^*$).  
The mean-field-type equilibrium state $\mathbf{z}^*$ is \textit{locally stable} if there exists $\delta>0$ such that for $\mu^X$-almost every $x$, for any initial condition $\tilde{\mathbf{z}}\in\mathcal{Z}$ with $d_{\mathcal{Z}}(\tilde{\mathbf{z}},\mathbf{z}^*)<\delta$, the trajectory
\[
\mathbf{z}_{t+1} = T_{\nu^*}(\mathbf{z}_t) \;:=\; \bigl( T_{\nu^*,1}(\mathbf{z}_t),\dots,T_{\nu^*,n}(\mathbf{z}_t) \bigr),\qquad \mathbf{z}_0 = \tilde{\mathbf{z}},
\]
satisfies
\[
\limsup_{t\to\infty} d_{\mathcal{Z}}(\mathbf{z}_t,\mathbf{z}^*) \;\le\; \varepsilon(\delta),
\]
with $\varepsilon(\delta)\to 0$ as $\delta\to 0$.
\end{enumerate}
\end{definition}

\begin{remark}
Compared to a standard Nash equilibrium, where equilibrium conditions enforce that each player’s action is a best response to others’ actions, the Mean-Field-Type Causal Mirage Equilibrium imposes a fundamentally different type of consistency. It is a deeper refinement.  Condition (i) enforces policy coherence by requiring actions to be fully induced by the internal representation, rather than independently optimized. Condition (ii) replaces strategic optimality with dynamical invariance, requiring the semantic state to be a fixed point of the closed-loop evolution induced jointly by the policy and the environment. Condition (iii) introduces a regime constraint absent in classical game theory: equilibrium is only admissible when the system operates in a high-reinforcement, low-grounding phase where internal reinforcement and confidence dominate causal alignment. Condition (iv) restricts equilibria to the invariant domain of the induced dynamics, ensuring that the representation-action pair is self-contained under iteration. Condition (v) strengthens the notion from a static fixed point to a dynamically stable attractor, requiring robustness under perturbations. Unlike standard Nash equilibrium which characterizes optimal strategic behavior in action space per player, CME characterizes the persistence of self-consistent but potentially causally detached semantic states under recursive generative dynamics (Figure \ref{fig:mirage_bifurcation_topology}).
\end{remark}

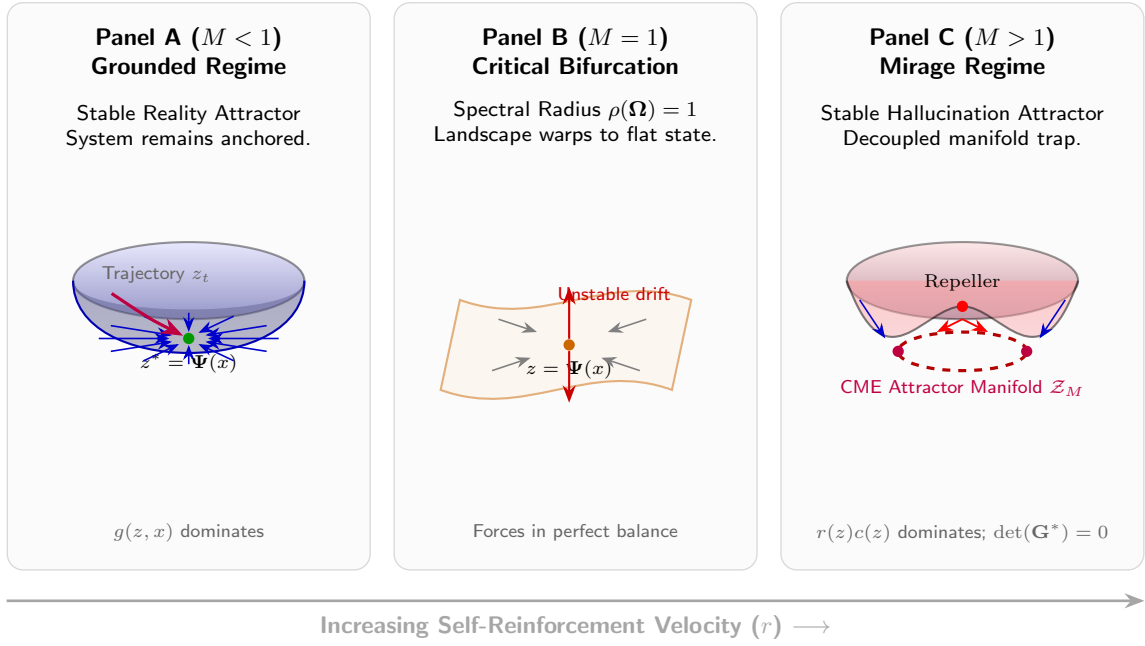
\begin{figure*}[htbp]
\centering
\begin{tikzpicture}[
    >=Stealth,
    panelbox/.style={draw=gray!30, fill=gray!3, rounded corners=6pt, line width=0.5pt, minimum height=7.5cm, minimum width=4.8cm, anchor=north west},
    titlefont/.style={font=\small\bfseries\sffamily, align=center},
    subfont/.style={font=\footnotesize\sffamily, align=center},
    labelfont/.style={font=\scriptsize\sffamily, color=gray!80!black}
]

    \node[panelbox] (panelA) at (0,0) {};
    \node[titlefont, below=0.2cm of panelA.north] (titleA) {\textbf{Panel A} ($M < 1$)\\Grounded Regime};
    \node[subfont, below=0.1cm of titleA] {Stable Reality Attractor\\System remains anchored.};

    \begin{scope}[shift={($(panelA.center)+(0,-1.2)$)}, scale=0.85]
        \draw[top color=blue!20, bottom color=blue!50!black, opacity=0.4, thick] (0,1.5) ellipse (1.8cm and 0.6cm);
        \fill[bottom color=blue!10, top color=blue!40, opacity=0.3] (-1.8,1.5) .. controls (-1.8,0) and (1.8,0) .. (1.8,1.5) -- (1.8,1.5) .. controls (1.8,0.9) and (-1.8,0.9) .. (-1.8,1.5);
        \draw[thick, color=blue!70!black] (-1.8,1.5) .. controls (-1.8,0) and (1.8,0) .. (1.8,1.5);
        
        \node[circle, fill=green!60!black, inner sep=1.5pt, label={below:\scriptsize $z^* = \mathbf{\Psi}(x)$}] (targetA) at (0,0.6) {};
        
        \foreach \angle in {0,30,60,90,120,150,180,210,240,270,300,330} {
            \draw[->, blue!80!black, line width=0.6pt] ({1.4*cos(\angle)}, {0.6 + 0.4*sin(\angle)}) -- ({0.3*cos(\angle)}, {0.6 + 0.1*sin(\angle)});
        }
        \draw[->, purple, thick, line width=1.2pt] (-1.2, 1.3) .. controls (-0.6, 0.9) .. (targetA);
        \node[labelfont, anchor=west] at (-1.5, 1.6) {Trajectory $z_t$};
    \end{scope}
    \node[labelfont, yshift=0.5cm] at (panelA.south) {$g(z,x)$ dominates};

    \node[panelbox, right=0.3cm of panelA] (panelB) {};
    \node[titlefont, below=0.2cm of panelB.north] (titleB) {\textbf{Panel B} ($M = 1$)\\Critical Bifurcation};
    \node[subfont, below=0.1cm of titleB] {Spectral Radius $\rho(\mathbf{\Omega}) = 1$\\Landscape warps to flat state.};

    \begin{scope}[shift={($(panelB.center)+(0,-1.2)$)}, scale=0.85]
        \draw[thick, color=orange!80!black, fill=orange!10, opacity=0.5] 
            (-1.8,1.2) to[out=-20,in=160] (1.8,1.2) -- 
            (1.5,-0.2) to[out=160,in=-20] (-2.1,-0.2) -- cycle;
            
        \node[circle, fill=orange!80!black, inner sep=1.5pt, label={below:\scriptsize $z = \mathbf{\Psi}(x)$}] (targetB) at (-0.1,0.5) {};
        
        \draw[->, gray, line width=0.6pt] (-1.2, 0.9) -- (-0.6, 0.7);
        \draw[->, gray, line width=0.6pt] (1.1, 0.9) -- (0.5, 0.7);
        \draw[->, gray, line width=0.6pt] (-1.3, 0.1) -- (-0.7, 0.3);
        \draw[->, gray, line width=0.6pt] (1.0, 0.1) -- (0.4, 0.3);
        
        \draw[->, red!80!black, thick] (targetB) -- (-0.1, 1.4);
        \draw[->, red!80!black, thick] (targetB) -- (-0.1, -0.4);
        \node[labelfont, color=red!80!black] at (0.6, 1.3) {Unstable drift};
    \end{scope}
    \node[labelfont, yshift=0.5cm] at (panelB.south) {Forces in perfect balance};

    \node[panelbox, right=0.3cm of panelB] (panelC) {};
    \node[titlefont, below=0.2cm of panelC.north] (titleC) {\textbf{Panel C} ($M > 1$)\\Mirage Regime};
    \node[subfont, below=0.1cm of titleC] {Stable Hallucination Attractor\\Decoupled manifold trap.};

    \begin{scope}[shift={($(panelC.center)+(0,-1.2)$)}, scale=0.85]
        \draw[top color=red!20, bottom color=purple!60!black, opacity=0.4, thick] (0,1.5) ellipse (1.8cm and 0.6cm);
        
        \shadedraw[bottom color=purple!20, top color=red!40, opacity=0.5, thick] 
            (-1.8,1.5) .. controls (-1.2,-0.2) and (-0.6,1.1) .. (0,1.1)
                       .. controls (0.6,1.1) and (1.2,-0.2) .. (1.8,1.5);
                       
        \draw[line width=1.2pt, red!70!black, dashed] (0,0.4) ellipse (1.0cm and 0.3cm);
        
        \node[circle, fill=red, inner sep=1.5pt, label={above:\scriptsize Repeller}] (targetC) at (0,1.1) {};
        
        \draw[->, red, line width=0.7pt] (0,0.9) -- (-0.4, 0.7);
        \draw[->, red, line width=0.7pt] (0,0.9) -- (0.4, 0.7);
        
        \node[circle, fill=purple, inner sep=1.5pt] (cme1) at (-1.0, 0.4) {};
        \node[circle, fill=purple, inner sep=1.5pt] (cme2) at (1.0, 0.4) {};
        \node[labelfont, color=purple, anchor=north] at (0, 0.1) {CME Attractor Manifold $\mathcal{Z}_M$};
        
        \draw[->, blue!80!black, line width=0.6pt] (-1.6, 1.2) -- (-1.2, 0.6);
        \draw[->, blue!80!black, line width=0.6pt] (1.6, 1.2) -- (1.2, 0.6);
    \end{scope}
    \node[labelfont, yshift=0.5cm] at (panelC.south) {$r(z)c(z)$ dominates; $\det(\mathbf{G}^*) = 0$};

    \begin{scope}[on background layer]
        \draw[->, line width=1pt, gray!70] ($(panelA.south west)+(0,-0.4)$) -- ($(panelC.south east)+(0,-0.4)$) 
            node[pos=0.5, below=0.05cm, font=\small\sffamily\bfseries] {Increasing Self-Reinforcement Velocity ($r$) $\longrightarrow$};
    \end{scope}

\end{tikzpicture}
\caption{Geometric mechanics of the Mirage Trap bifurcation in the latent semantic space $\mathcal{Z}$. As the dimensionless Mirage Intensity crosses the critical parameter boundary ($M=1$), the perfectly grounded reality representation $\mathbf{\Psi}(x)$ transitions from a unique globally stable attractor into an unstable repeller. The system trajectories smoothly escape into the invariant ungrounded topological valley representing the stable Mean-Field-Type Causal Mirage Equilibrium manifold.}
\label{fig:mirage_bifurcation_topology}
\end{figure*}

We establish sufficient conditions for the existence of a Multi-Agent  Mean-Field-Type Causal Mirage Equilibrium as defined in Definition~\ref{def:macme}. The proof relies on a fixed‑point argument in the space of probability measures, using the Kakutani-Glicksberg-Fan theorem.

Let $(\mathcal{X},d_{\mathcal{X}})$, $(\mathcal{Z},d_{\mathcal{Z}})\subset\mathbb{R}^d$, and $(\mathcal{A}_i,d_{\mathcal{A}_i})$ be compact metric spaces. Endow each $\mathcal{P}(\cdot)$ with the weak topology, making them compact metrizable spaces (Prokhorov's theorem). Denote
\[
\mathcal{M} := \mathcal{P}(\mathcal{Z} \times \mathcal{A} \times \mathcal{X}),\ 
\mathcal{M}^{\mathcal{Z}} := \mathcal{P}(\mathcal{Z}),\ 
\mathcal{M}^{\mathcal{Z}\times\mathcal{X}} := \mathcal{P}(\mathcal{Z}\times\mathcal{X}).
\]

We make the following regularity hypotheses:

\begin{enumerate}[label=(H\arabic*)]
\item
    \begin{itemize}
        \item $\pi_i : \mathcal{Z} \times \mathcal{M}^{\mathcal{Z}} \to \mathcal{A}_i$ is continuous.
        \item $\Phi_i : \mathcal{Z} \times \mathcal{A} \times \mathcal{X} \times \mathcal{M} \to \mathcal{Z}$ is continuous.
        \item $u_i : \mathcal{X} \times \mathcal{Z} \times \mathcal{A}_i \times \mathcal{A}_{-i} \times \mathcal{M} \to \mathbb{R}$ is continuous.
        \item $r_i, c_i : \mathcal{Z} \times \mathcal{M}^{\mathcal{Z}} \to [0,\infty)$ are continuous.
        \item $g_i : \mathcal{X} \times \mathcal{X} \times \mathcal{M}^{\mathcal{Z}\times\mathcal{X}} \to [0,1]$ is continuous.
    \end{itemize}
\item Each $\mathcal{A}_i$ is a non-empty convex compact subset of a Euclidean space (or we work with mixed strategies $\Delta(\mathcal{A}_i)$, but here we keep pure actions and assume $\mathcal{A}_i$ convex).
\item  For every fixed $x,\mathbf{z},\mathbf{a}_{-i},\nu$, the map $a_i \mapsto u_i(x,\mathbf{z},a_i,\mathbf{a}_{-i},\nu)$ is concave.
\item The $\tau_i$-expectile $e_{\tau_i}(U)$ is continuous in the law of $U$ (under weak convergence) and strictly quasi‑concave in the action when the action influences the distribution linearly. This holds, e.g., if $u_i$ is linear in the measure $\nu^{X,\mathbf{Z},a_i',\mathbf{a}_{-i}}$, which is itself linear in $a_i$ when $\nu$ is fixed.
\item For each $i$, the correspondence
    \[
    F_i(\mathbf{z},\mathbf{a}_{-i},x,\nu) := \{ a_i \in \mathcal{A}_i : M_i(\mathbf{z}^{(i)},x,\nu^{\mathcal{Z}\times\mathcal{X}}) \ge 1+\eta_i \}
    \]
    is nonempty, convex, and continuous (upper and lower hemicontinuous) in $(\mathbf{z},\mathbf{a}_{-i},x,\nu)$.
\item \textbf{Measure compactness:} $\mathcal{M}$ is compact in the weak topology.
\end{enumerate}

\begin{theorem}[Existence of CME]\label{thm:existence}
Under assumptions (H1)-(H6), there exists a tuple $(\mathbf{z}^*,\mathbf{a}^*,\nu^*)$ that satisfies the conditions of Definition~\ref{def:macme} :  Mean-Field-Type Causal Mirage Equilibrium exists.
\end{theorem}

\begin{proof}

The proof proceeds by constructing a fixed point of a suitable correspondence defined on the compact convex set $\mathcal{M}$ (which is a convex subset of the vector space of finite signed measures). The key idea: associate to each candidate joint law $\nu$ the set of all joint laws that are stationary under the best‑reply dynamics constrained by mirage feasibility, and then enforce the semantic fixed‑point condition.

For a fixed $\nu \in \mathcal{M}$ with $\nu^{\mathcal{X}}=\mu^X$, define for each agent $i$ the set of admissible actions given state $(\mathbf{z},\mathbf{a}_{-i},x)$:
$$
B_i(\mathbf{z},\mathbf{a}_{-i},x,\nu) $$ $$:= \arg\max_{a_i \in F_i(\mathbf{z},\mathbf{a}_{-i},x,\nu)} e_{\tau_i}\bigl( u_i(x,\mathbf{z},a_i,\mathbf{a}_{-i},\nu^{X,\mathbf{Z},a_i,\mathbf{a}_{-i}}) \bigr).
$$
Because $F_i$ is nonempty, convex, and continuous, and the objective is continuous and strictly quasi‑concave in $a_i$ (by (H4)), the argmax is a nonempty, convex, compact-valued correspondence that is upper hemicontinuous (Berge's theorem). Define the product best‑response
\[
B(\mathbf{z},\mathbf{a}_{-1},\dots,\mathbf{a}_{-n},x,\nu) := \prod_{i=1}^n B_i(\mathbf{z},\mathbf{a}_{-i},x,\nu).
\]

Given $\nu$ and a measurable selection $\mathbf{a}' \in B(\mathbf{z},\mathbf{a}_{-1},\dots,\mathbf{a}_{-n},x,\nu)$, define the next semantic state
\[
\mathbf{z}' = \Phi(\mathbf{z},\mathbf{a}',x,\nu) \in \mathcal{Z}.
\]
This defines a measurable map $T_{\nu,x} : \mathcal{Z} \times \mathcal{A} \to \mathcal{Z}$.

Define a correspondence $\Gamma: \mathcal{M} \to 2^{\mathcal{M}}$ as follows: $\Gamma(\nu)$ consists of all joint laws $\nu' \in \mathcal{M}$ such that
\begin{itemize}
    \item $(\nu')^{\mathcal{X}} = \mu^X$,
    \item there exists a measurable selection $\mathbf{a}'(\mathbf{z},\mathbf{a},x)$ from $B(\mathbf{z},\mathbf{a}_{-1},\dots,\mathbf{a}_{-n},x,\nu)$ satisfying
    \[
    \nu' = (\Phi(\mathbf{Z},\mathbf{a}'(\mathbf{Z},\mathbf{A},X),\nu),\; \mathbf{a}'(\mathbf{Z},\mathbf{A},X),\; X)_{\#} \nu,
    \]
    i.e., $\nu'$ is the pushforward of $\nu$ under $(\mathbf{z},\mathbf{a},x) \mapsto (\Phi(\mathbf{z},\mathbf{a}',x,\nu),\, \mathbf{a}',\, x)$.
\end{itemize}
In words, $\nu'$ is the one‑step evolution of the joint law when every agent plays a mirage‑feasible best response.

We require that the semantic fixed‑point condition holds: for $\nu'$-almost every $(\mathbf{z}',\mathbf{a}',x)$ we have $\mathbf{z}' = \Phi(\mathbf{z}',\mathbf{a}',x,\nu')$. This is a consistency condition that can be incorporated by restricting $\Gamma(\nu)$ to those $\nu'$ that are invariant under the map induced by the same $\nu'$ (self‑consistency). To handle this, we consider the set of fixed points of the composite mapping.

We define a second correspondence $\Lambda: \mathcal{M} \to 2^{\mathcal{M}}$ where $\Lambda(\nu)$ consists of all $\nu'$ that satisfy the stationarity condition (ii) and the semantic fixed point (iii) simultaneously, given that actions are chosen from $B$ computed with $\nu$. By construction, any fixed point $\nu^* \in \Gamma(\nu^*)$ will satisfy conditions (i)–(iv) of Definition~\ref{def:macme}. The local stability condition (v) is not required for existence; it is a separate property that can be verified a posteriori.

\begin{itemize}
    \item \textit{Nonemptiness:} For each $\nu$, $F_i$ is nonempty, $B_i$ is nonempty, and the pushforward of $\nu$ under any measurable selection yields a probability measure $\nu'$; thus $\Gamma(\nu) \neq \emptyset$.
    \item \textit{Convexity:} $\mathcal{M}$ is convex. If $\nu_1',\nu_2' \in \Gamma(\nu)$, then any convex combination $\lambda \nu_1' + (1-\lambda)\nu_2'$ is also a pushforward of $\nu$ under a suitable mixture of selections; because the dynamics are linear in the measure and selections can be randomized, the set of achievable $\nu'$ is convex.
    \item \textit{Closed graph:} By the continuity assumptions (H1), the boundedness of utilities, and the compactness of $\mathcal{M}$, the correspondence $\Gamma$ has a closed graph in $\mathcal{M} \times \mathcal{M}$. This follows from the fact that if $\nu_k \to \nu$ and $\nu_k' \in \Gamma(\nu_k)$ with $\nu_k' \to \nu'$, then the measurable selections converge to a selection for $\nu$ (using the measurable maximum theorem and weak convergence).
    \item \textit{Compactness:} $\mathcal{M}$ is compact, so the image sets $\Gamma(\nu)$ are subsets of a compact set; therefore $\Gamma$ is compact‑valued.
\end{itemize}
Thus $\Gamma$ is a nonempty, convex‑valued, upper hemicontinuous correspondence from the compact convex set $\mathcal{M}$ into itself.

The Kakutani-Glicksberg-Fan fixed‑point theorem (for locally convex topological vector spaces) applies to the correspondence $\Gamma$ on the compact convex set $\mathcal{M}$ (which is a subset of the space of signed measures with the weak-* topology). Hence there exists $\nu^* \in \mathcal{M}$ such that $\nu^* \in \Gamma(\nu^*)$.

From $\nu^* \in \Gamma(\nu^*)$ we obtain:
\begin{itemize}
    \item The stationarity condition (ii) holds by construction.
    \item The semantic fixed point (iii) holds because $\nu^*$ is a fixed point of the pushforward, meaning that for $\nu^*$-almost every $(\mathbf{z}^*,\mathbf{a}^*,x)$ we have $\mathbf{z}^* = \Phi(\mathbf{z}^*,\mathbf{a}^*,x,\nu^*)$. (If the pushforward equality holds, then the identity mapping on $\mathcal{Z}$ is a version of the transition, which forces $\mathbf{z}^*$ to be a fixed point.)
    \item The policy consistency (i) and the mirage feasibility are embedded in the definition of $B_i$; since $a_i^*$ is a selection from $B_i$, we have $a_i^* \in F_i$ and therefore $M_i(\mathbf{z}^*,x,(\nu^*)^{\mathcal{Z}\times\mathcal{X}}) \ge 1+\eta_i > 1$, so (iv) holds.
\end{itemize}
Thus $(\mathbf{z}^*,\mathbf{a}^*,\nu^*)$ is a Multi‑Agent Mean-Field-Type Causal Mirage Equilibrium. The local stability condition (v) is not automatically guaranteed and requires additional assumptions (e.g., contractivity of $T_{\nu^*}$ near $\mathbf{z}^*$), which can be imposed separately.
If in addition the dynamics $T_{\nu^*}$ are a contraction on a neighborhood of $\mathbf{z}^*$ (e.g., $d(T_{\nu^*}(\mathbf{z}),\mathbf{z}^*) \le \kappa d(\mathbf{z},\mathbf{z}^*)$ with $\kappa<1$), then $\mathbf{z}^*$ is locally exponentially stable and condition (v) holds. Such contractivity can be ensured by spectral conditions on the Jacobian of $\Phi$ at equilibrium.

\noindent This completes the proof of Theorem~\ref{thm:existence}.
\end{proof}

\subsection{Stationary Risk-Sensitive Constrained Mean-Field-Type Game}
\label{sec:CME_MFTG}

We now reinterpret the Causal Mirage Equilibrium as the equilibrium solution of a stationary constrained risk-sensitive mean-field-type game \cite{audioiamali2,audioiamali}.

Let
$$
\mathcal G^{\rm CME} $$  $$
=
\Big(
\mathcal I,
(\mathcal A_i)_{i\in\mathcal I},
(\Phi_i)_{i\in\mathcal I},
(u_i)_{i\in\mathcal I},
(g_i,r_i,c_i)_{i\in\mathcal I},
(\tau_i,\eta_i)_{i\in\mathcal I},
\mu^X
\Big)
$$
be the game induced by the primitives introduced above.

The player set is
$
\mathcal I=\{1,\ldots,n\},
\  n\ge2.
$

For each player $i,$ the action space is $\mathcal A_i$, the semantic state space is $\mathcal Z_i$, and the common external state space is $\mathcal X$. Denote
$
\mathcal Z=\prod_{j=1}^n\mathcal Z_j,
\ 
\mathcal A=\prod_{j=1}^n\mathcal A_j.
$

The stationary state-action of the system is
$
\omega=(x,\mathbf z,\mathbf a)
\in
\mathcal X\times\mathcal Z\times\mathcal A.
$

Player $i$ observes its information variable 
$
I_i=(x,\mathbf z),
$
and chooses an action according to a stationary measurable policy
$
\pi_i: \ \mathcal X\times\mathcal Z \times \mathcal{P}({\mathcal{Z}\times\mathcal{X}}) \rightarrow\mathcal A_i.
$

The collection
$
\pi=(\pi_i)_{i\in\mathcal I}
$
induces a stationary joint probability measure
$
\nu^\pi
\in
\mathcal P(\mathcal X\times\mathcal Z\times\mathcal A),
$
where
$
a_i=\pi_i(x,\mathbf z, \nu^{\mathcal{Z}\times\mathcal{X}}),
\ 
z_i=\Phi_i(x,\mathbf z,\mathbf a,\nu^\pi).
$

The probability law induced by the joint stationary terms is 
$
\nu^\pi
=
\Big(
X,
\Phi(X,\mathbf Z,\mathbf A,\nu^\pi),
\mathbf A
\Big)_{\#}\nu^\pi .
$

For each player $i$, admissible actions are restricted by the mirage constraint
$
M_i(\mathbf z^{(i)},x,\nu^\pi)
\ge
1+\eta_i,
$
where
$
z_i^{(i)}
=
\Phi_i(x,\mathbf z,(a_i,\mathbf a_{-i}),\nu^\pi).
$

The resulting feasible-action correspondence is
$
F_i(\mathbf z,\mathbf a_{-i},x,\nu^\pi)
= \Big\{
a_i\in\mathcal A_i:
M_i(\mathbf z^{(i)},x,\nu^\pi)\ge1+\eta_i
\Big\}.
$

Given $(\mathbf z,\mathbf a_{-i},x,\nu^\pi)$, player $i$ evaluates actions through the expectile risk criterion
$
a_i
\mapsto
e_{\tau_i}
\left(
u_i
(x,\mathbf z,a_i,\mathbf a_{-i},\nu^\pi)
\right),
$
and solves
$
\sup_{a_i\in F_i(\mathbf z,\mathbf a_{-i},x,\nu^\pi)}
e_{\tau_i}
\left(
u_i
(x,\mathbf z,a_i,\mathbf a_{-i},\nu^\pi)
\right).
$

The associated best-response correspondence is
$$
BR_i(\mathbf z,\mathbf a_{-i},x,\nu^\pi)
=
\arg\max_{a_i\in F_i(\mathbf z,\mathbf a_{-i},x,\nu^\pi)}
e_{\tau_i} \left(
u_i
(x,\mathbf z,a_i,\mathbf a_{-i},\nu^\pi)
\right).
$$

\begin{theorem}[CME as a Stationary  Constrained MFTG Equilibrium]
\label{thm:CME_MFTG_equivalence}
Let the primitives satisfy the assumptions of Theorem~\ref{thm:existence}. A triple
$
(\mathbf z^*,\mathbf a^*,\nu^*)
$
is a Causal Mirage Equilibrium in the sense of Definition~\ref{def:macme} if and only if the following conditions hold:

\begin{enumerate}
\item
$
a_i^*
\in
BR_i
(\mathbf z^*,\mathbf a_{-i}^*,x,\nu^*),
\qquad
\forall i\in\mathcal I;
$

\item
$
z_i^*
=
\Phi_i
(\mathbf z^*,\mathbf a^*,x,\nu^*),
\qquad
\forall i\in\mathcal I;
$

\item
$
\nu^*
=
\Big(
X,
\Phi(X,\mathbf Z,\mathbf A,\nu^*),
\mathbf A
\Big)_{\#}\nu^*;
$

\item
$
M_i(\mathbf z^*,x,\nu^*)
\ge
1+\eta_i,
\qquad
\forall i\in\mathcal I.
$
\end{enumerate}

Hence the Causal Mirage Equilibrium coincides exactly with the stationary equilibrium of the constrained risk-sensitive mean-field-type game $\mathcal G^{\rm CME}$.
\end{theorem}

\section{Conclusion and Future Work}
\label{sec:conclusion}

We introduced CME, a new equilibrium concept that captures stable hallucination dynamics in generative machine intelligence. CME formalizes how semantic representations can become self‑sustaining despite losing causal grounding. We proved existence under mild conditions, established a bifurcation theorem, defined epistemic stability, and extended the framework to multi‑agent populations. The theory provides a mathematical foundation for analyzing hallucinations, confabulations, and recursive misinformation. Future directions include: (i) stochastic CME with random external states; (ii) mean-field-type games where agents strategically manipulate mirage intensity; (iii) MFTG design  for grounding maintenance; (iv) empirical estimation of mirage intensity in large language models; (v) thermodynamic interpretations of semantic reinforcement as entropy reduction.

\section*{Acknowledgments}
We thank the anonymous reviewers for their insightful comments.

\bibliographystyle{plain}

\end{document}